\documentclass[11pt,a4paper]{article}

\usepackage{indentfirst,mathrsfs}
\usepackage{amsfonts,amsmath,amssymb,amsthm}
\usepackage{latexsym,amscd}
\usepackage{amsbsy}
\usepackage{color}
\usepackage[noadjust]{cite}
\usepackage{cite}
\usepackage{multirow}
\usepackage{makecell}
\usepackage{float}
\usepackage{cases}
\usepackage{diagbox}
\usepackage{threeparttable}

\newtheorem{thm}{Theorem}
\newtheorem{lem}{Lemma}
\newtheorem{cor}{Corollary}

\newtheorem{prop}{Proposition}

\newtheorem{defn}{Definition}
\newtheorem{example}{Example}
\newtheorem{remark}{Remark}

\newcommand{\F}{\mathbb{F}}
\newcommand{\ftwon}{{\mathbb F}_{2^n}}
\newcommand{\ftwom}{{\mathbb F}_{2^m}}

\newcommand{\bv}{\overline{v}}

\begin{document}
\title{Construction of $(n,n)$-functions with low differential-linear uniformity}
\author{Xi Xie, Nian Li, Qiang Wang, Xiangyong Zeng, Yinglong Du
\thanks{X. Xie, N. Li and Y. Du are with the Key Laboratory of Intelligent Sensing System and Security (Hubei University), Ministry of Education, the Hubei Provincial Engineering Research Center of Intelligent Connected Vehicle Network Security, and School of Cyber Science and Technology, Hubei University, Wuhan 430062, China.
Q. Wang is with the School of Mathematics and Statistics, Carleton University, Ottawa, K1S 5B6, Canada.
X. Zeng is with the Key Laboratory of Intelligent Sensing System and Security (Hubei University), Ministry of Education, Hubei Key Laboratory of Applied Mathematics, and Faculty of Mathematics and Statistics, Hubei University, Wuhan 430062, China.
Email: xi.xie@aliyun.com, nian.li@hubu.edu.cn, wang@math.carleton.ca, xiangyongzeng@aliyun.com, YingLong.Du@aliyun.com.
}
}
	\date{}
	\maketitle
\begin{quote}
{{\bf Abstract:}
	The differential-linear connectivity table (DLCT), introduced by Bar-On et al. at EUROCRYPT'19, is a novel tool that captures the dependency between the two subciphers involved in differential-linear attacks.
	This paper is devoted to exploring the differential-linear properties of $(n,n)$-functions. First, by refining specific exponential sums, we propose two classes of power functions over $\F_{2^n}$ with low differential-linear uniformity (DLU).
	Next, we further investigate the differential-linear properties of $(n,n)$-functions that are polynomials by utilizing power functions with known DLU.
	Specifically, by combining a cubic function with quadratic functions, and employing generalized cyclotomic mappings, we construct several classes of $(n,n)$-functions with low DLU, including some that achieve optimal or near-optimal DLU compared to existing results.
		}	
	
	\end{quote}

\section{Introduction}
The substitution box (S-box) is mathematically defined as a vectorial Boolean function from the vector space $\F_{2}^n$ to $\F_{2}^m$, also referred to as an $(n,m)$-function, where $n$ and $m$ are positive integers. Since the vector space $\F_{2}^n$ is isomorphic to the finite field $\F_{2^n}$, such a function can also be viewed as a mapping from $\F_{2^n}$ to $\F_{2^m}$. The security of a cipher heavily depends on the cryptographic properties of its S-box, as it is typically the only nonlinear component in most modern block ciphers. Therefore, an S-box used in cryptography should possess good properties to resist various attacks.

Differential \cite{BS} and linear cryptanalysis \cite{Matsui} are powerful techniques for assessing the security of block ciphers. In 1994, Langford and Hellman published a combination of differential and linear cryptanalysis under two default independence assumptions, known as differential-linear cryptanalysis \cite{LH}.
To delve deeper into the resistance of S-boxes against differential-linear attacks, Bar-On et al. \cite{Bar-On} introduced a new tool called the differential-linear connectivity table ({\rm DLCT}) in EUROCRYPT'19. This tool allows for precise analysis of the probabilities of the connection between the differential and linear parts in a differential-linear attack. Later, Li et al. \cite{LLLQ} and Canteaut, K\"{o}lsch, and Wiemer \cite{Canteaut1} independently provided theoretical characterizations of the {\rm DLCT} for cryptographic functions, and their results were subsequently unified in \cite{Canteaut2}.
Specifically, they observed that the DLCT aligns with the autocorrelation of vectorial Boolean functions and analyzed the differential-linear uniformity (DLU) of several classes of functions, some of which were derived from known autocorrelation results \cite{Charpin,Gong,Carlet,SW,CanteautC}. Building on this foundation, subsequent researchers explored the constructions of balanced vectorial Boolean functions with low DLU \cite{Tang1,Tang2}. 
Subsequently, Jeong, Koo, and Kwon \cite{Jeong} demonstrated that two classes of differentially 4-uniform permutations, constructed from the inverse function composed by disjoint cycles, exhibit low DLU.
To date, only a few functions with known DLU have been identified, most of which are $(n,n)$-functions. In this context, $(n,n)$-functions refer to mappings  from $\F_{2^n}$ to  $\F_{2^n}$.
Table \ref{DLU-table} summarizes the infinite families of $(n,n)$-functions with known DLU, while further studies on $(n,m)$-functions can be found in \cite{Tang1,Tang2}.

\begin{table}[!htb]\footnotesize
\begin{threeparttable}
	\caption{The $(n,n)$-functions $F(x)$ with known {\rm DLU}}  \label{DLU-table}
	\renewcommand\arraystretch{1.2}
	\setlength\tabcolsep{2pt}
	\centering
	\begin{tabular}{|c|c|c|c|c|}
		\hline No.&  $F(x)$                & Condition                   &  ${\rm DLU}_{F}$    & Refs.                  \\ \hline\hline
		1  &  $x^{2^n-2}$            &    $n$ even           &    $2^{n/2}$       &\cite{Charpin}            \\ \hline
		2  &  $x^{2^{2k}-2^k+1}$     & $3k\equiv \pm 1({\rm mod}\, n)$, $n$ odd    &    $2^{(n-1)/2}$        &\cite{Gong}            \\ \hline
		3  &  $x^{2^{(n+1)/2}+3}$     &  $n$ odd    &    $2^{(n-1)/2}$ or  $2^{(n+1)/2}$      &\cite{Carlet}            \\ \hline
		4  &  $x^{2^{(n-1)/2}+3}$            &$n$ odd                   & ${\rm DLU}_{F}\leq2^{(n+3)/2}$    &\cite{Carlet}            \\ \hline
		5  &  $x^{2^{m+1}+3}$            &$n=2m$                    & ${\rm DLU}_{F}\leq2^{3m/2}$    &\cite{SW}            \\ \hline
		6  & $x^{2^m+2^{(m+1)/2}+1}$      &$n=2m$, $m$ odd  & ${\rm DLU}_{F}\leq2^{3m/2}$              &\cite{SW}       \\ \hline
		7  & $x^{2^{2k}+2^k+1}$             &$n=6k$                    & $2^{5k-1}$         &\cite{CanteautC}             \\ \hline
		8  &$x^{2^{2k}+2^k+1}$            &$n=4k$    & $2^{3k-1}$   &\cite{Canteaut2}  \\ \hline
		9  & $\sum_{0\leq i<j\leq n-1} a_{ij}x^{2^i+2^j}$             & $n>0$  & $2^{n-1}$  &\cite{Canteaut2}             \\ \hline
10  & $Inv\circ (0,1)$             & $n\geq 4$ even  & $2^{n/2}$  &\cite{Jeong}             \\ \hline
11 & $Inv\circ (0,1)(\alpha,\beta)$    & $n$ even  & $2^{n/2}+4$  &\cite{Jeong}             \\ \hline
		12  & $x^{2^{2k}+2^{k}+1}$ & $\gcd(k,\,n)=e$                  & ${\rm DLU}_{F}\!\leq\!2^{(n+3e)/2-1}$\!\,or\,$2^{(n+4e)/2-1}$     & This paper            \\ \hline
		13  & $x^{l(2^m-1)}$ &$n\!=\!2m,\gcd(l,2^m\!+\!1)\!=\!1$        & $\frac{K_{\max}(m)^2}{2}$\!\,or\,$\frac{K_{\max}(m)^2}{2}\!+\!2K_{\max}(m)$          & This paper             \\ \hline
		14 & $x^{2^{2k}+2^k+1}+\!\!\!\!\sum\limits_{0\leq i<j\leq n-1}\!\!\!\! a_{ij}x^{2^i+2^j}$         & $\gcd(k,\,n)=e$                  & ${\rm DLU}_{F}\!\leq\!2^{(n+3e)/2-1}$\!\,or\,$2^{(n+4e)/2-1}$ &   This paper   \\\hline
		15 & $x^{2^n-2}$ if $x\in \F_{2^n}\backslash \{\xi\}$; $a$ if $x\!=\!\xi$
		 & $n$ even &  $2^{n/2}$\,or\,$2^{n/2}+2$ & This paper \\\hline
		16 & $x^{2^{2k}\!-2^k+1}$ if $x\!\in\! \F_{2^n}\!\!\backslash\! \{\xi\}$; $a$ if $x\!=\!\xi$
		 & $n$ odd & ${\rm DLU}_{F}\leq2^{(n-1)/2}+2$ & This paper \\\hline
	\end{tabular}
-where $K_{\max}(m)$ is given by \eqref{Kmax}, $Inv$ denotes the inverse function, and $(0,1)$, $(\alpha,\beta)$ are transpositions over $\F_{2^n}$.
\end{threeparttable}
\end{table}
This paper aims to construct new infinite families of $(n,n)$-functions with low {\rm DLU}.
Power functions are preferred candidates for S-boxes due to their simple algebraic forms and typically lower hardware implementation costs.
Most existing research on the {\rm DLU} of $(n,n)$-functions has focused on power functions.
As noted by Canteaut et al. \cite{Canteaut2}, the {\rm DLU} of an $(n,n)$-function is greater than $2^{n/2-1}$, and this bound can be further refined based on experimental results from Magma \cite{BCP}. Among the known results, the inverse function and Kasami APN permutation achieve the lowest {\rm DLU} for even and odd $n$, respectively, serving as benchmarks for evaluating the {\rm DLU} of new constructions.
In the first part of our study, we examine the {\rm DLU} of two classes of power functions that exhibit relatively low {\rm DLU} in the experimental data. Through refined manipulations of specific exponential sums, we derive upper bounds for their {\rm DLU}, leading to several classes of functions with low {\rm DLU}.
Building upon this, we extend our investigation to the {\rm DLU} of several classes of polynomials over $\F_{2^n}$. First, we propose a class of polynomials with lower {\rm DLU} by combining a cubic function and general quadratic functions.
Second, we investigate the {\rm DLU} of a class of generalized cyclotomic mappings, which essentially involve modifying the values of a power function at specific points or sets. By utilizing the inverse function and the Kasami APN permutation, we construct several classes of polynomials whose {\rm DLU} is either equal to or very close to theirs.

The rest of this paper is organized as follows. Section \ref{prel} introduces the preliminaries. Section \ref{cons1} investigates the {\rm DLU} of two classes of power functions and gives upper bounds on their {\rm DLU}. Section \ref{cons2} constructs two classes of $(n,n)$-functions with low {\rm DLU}, including polynomials achieving optimal or near-optimal {\rm DLU} compared to known results.
Section \ref{conc} concludes this study.

\section{Preliminaries}\label{prel}
Throughout this paper, $\#E$ denotes the cardinality of a finite set $E$. In addition, let $n$ be a positive integer and $\mathbb F_{2^n}$ be the finite field of $2^n$ elements. We denote by $\mathbb F_{2^n}^*$ the multiplicative cyclic group of non-zero elements of $\mathbb F_{2^n}$. The (absolute) trace function ${\rm Tr}_1^n:\mathbb F_{2^n}\longrightarrow \mathbb F_2$ is defined by ${\rm Tr}_1^n(x)=\sum_{i=0}^{n-1} x^{2^{i}}$ for all $x\in\mathbb F_{2^n}$.

\subsection{Exponential sums}

For each $b\in\mathbb{F}_{2^n}$, the function $\chi_b(x)=(-1)^{{\rm Tr}_1^n(bx)}$ defines an additive character
for $x\in\mathbb{F}_{2^n}$. The character $\chi_1$ is referred to as the canonical additive character of $\mathbb{F}_{2^n}$. For simplicity, we denote the canonical additive character of the prime field $\mathbb{F}_{2}$ by $\chi(x)=(-1)^{x}$ for $x\in\mathbb{F}_{2}$.

Below, we recall the classical binary Kloosterman sums and some results that are used in this paper. For any $\gamma\in\F_{2^n}$, the Kloosterman sum $K_{n}(\gamma)$ over $\F_{2^n}$ is defined as
\[K_{n}(\gamma)=\sum_{x\in\F_{2^n}}\chi({\rm Tr}_1^n(\gamma x+ x^{-1})).\]
The value of $K_{n}(\gamma)$ has been characterized as follows.
\begin{lem}{\rm(\cite{LW})}\label{lem.Kn1}
	Let $n\geq 3$ be a positive integer. Then for any integer $s\equiv 0 \, ({\rm mod} \, 4)$ in the range
	\[ [1-2^{(n+2)/2},\,1+2^{(n+2)/2}],\]
	there is an element $\gamma\in\ftwon$ such that $K_{n}(\gamma)=s$.
\end{lem}

\begin{lem}{\rm(\cite{Charpin,Helleseth,Lis})}\label{lem.Kn2}
	Let $n\geq 3$ be a positive integer and $\gamma\in\ftwon$. Then  $K_{n}(\gamma)\equiv 0\, ({\rm mod}\, 8)$ if ${\rm Tr}_1^n(\gamma)=0$ and $K_{n}(\gamma)\equiv 4\, ({\rm mod}\, 8)$ if ${\rm  Tr}_1^n(\gamma)=1$.
\end{lem}

Furthermore, we state the maximum and minimum values of a Kloosterman sum as follows.
\begin{lem}\label{lem.Kn1-2}
	Let $K_{\max}(n)$ and $K_{\min}(n)$ denote the maximum and minimum values of $K_n(\gamma)$ as $\gamma$ varies over $\ftwon$. Then
{\begin{equation}\label{Kmax}
		K_{\max}(n)=\left \{\begin{array}{lll}
			\lfloor2^{n/2+1}\rfloor+1, &{\rm if}\,\, \lfloor2^{n/2+1}\rfloor\equiv 3\,({\rm mod}\, 4);\\[0.05in]
			\lfloor2^{n/2+1}\rfloor-j, &{\rm if} \,\, \lfloor2^{n/2+1}\rfloor\equiv j\,({\rm mod}\, 4),\,\,j=0,1,2,
		\end{array}\right .
	\end{equation}
and
	\begin{equation*}
		K_{\min}(n)=\left \{\begin{array}{lll}
			4-\lfloor2^{n/2+1}\rfloor, &{\rm if}\,\, \lfloor2^{n/2+1}\rfloor\equiv 0\,({\rm mod}\, 4);\\[0.05in]
			j-\lfloor2^{n/2+1}\rfloor, &{\rm if} \,\, \lfloor2^{n/2+1}\rfloor\equiv j\,({\rm mod}\, 4),\,\,j=1,2,3.
		\end{array}\right .
	\end{equation*}}
\end{lem}

Let $n=2m$. For simplicity, denote the conjugate of $x\in\ftwon$ over $\ftwom$ by $\overline{x}$, i.e., $\overline{x}=x^{2^m}$. The unit circle of $\ftwon$ is defined as follows:
$$\mu_{2^m+1}:=\{z\in\ftwon: z\overline{z}=1\}.$$
The connection between the set $\mu_{2^m+1}$ and the Kloosterman sum is given by the following lemma, which will be used to prove our main result later.

\begin{lem}{\rm(\cite{D,LW,L})}\label{lem.km-U}
	Let $n=2m$ and $\gamma\in\ftwon^*$, where $m$ is a positive integer.
	Then
	\[\sum_{x\in\mu_{2^m+1}}\chi({\rm Tr}_1^n(\gamma x))=1-K_m(\gamma\overline{\gamma}).\]
\end{lem}

Note that $\mu_{2^m+1}\cap\ftwom=\{1\}$. It is well-known that each $x\in\ftwon^*$ can be uniquely written as $x=yz$ for some $y\in\ftwom^*$ and $z\in\mu_{2^m+1}$. Below, we provide an alternative expression for $x\in\ftwon\backslash\ftwom$, which will be useful in the subsequent computation of certain exponential sums.

\begin{lem}\label{lem.x}
	Let $n=2m$ be a positive integer. Then each element $x\in\ftwon\backslash\ftwom$ can be uniquely written as $x=v_1(v_2+1)/(v_1+v_2)$, where $v_1,v_2\in\mu_{2^m+1}\backslash \{1\}$ and $v_1\ne v_2$.
\end{lem}

\begin{proof}
	Let $x=v_1(v_2+1)/(v_1+v_2)$, where $v_1,v_2\in\mu_{2^m+1}\backslash \{1\}$ and $v_1\ne v_2$. Then we have $x\in\ftwon\backslash\ftwom$. Otherwise $x\in\ftwom$, i.e., $x^{2^m}=x$ gives \[\frac{v_1^{-1}(v_2^{-1}+1)}{v_1^{-1}+v_2^{-1}}=\frac{v_1(v_2+1)}{v_1+v_2}.\]
This gives  $(v_1+v_2)(1+v_1)(1+v_2)=0$, a contradiction.	
Thus $x\in\ftwon\backslash\ftwom$ for all $v_1\ne v_2\in\mu_{2^m+1}\backslash \{1\}$. Denote $\mathbb{V}:=\{(v_1,\,v_2): v_1,\,v_2\in\mu_{2^m+1}\backslash \{1\} \,\, {\rm and} \,\, v_1\ne v_2\}$.
	Next we claim that the mapping $\varphi: (v_1,\,v_2)\rightarrow \frac{v_1(v_2+1)}{v_1+v_2}$ from $\mathbb{V}$ to $\ftwon\backslash\ftwom$ is a bijection.
	Assume that for $(v_1,\,v_2)\ne (v_1',\,v_2')\in\mathbb{V}$, one has
	\begin{equation}\label{v1v2}
		\frac{v_1(v_2+1)}{v_1+v_2}=\frac{v_1'(v_2'+1)}{v_1'+v_2'},
	\end{equation}
which is equivalent to
\[\frac{v_1^2(v_2^2+1)}{v_1^2+v_2^2}=\frac{v_1'^2(v_2'^2+1)}{v_1'^2+v_2'^2}.\]
	Then we have $u_1/u_1'=v_1'/v_1$, where $u_1=v_1(v_2^2+1)/(v_1^2+v_2^2)$ and $u_1'=v_1'(v_2'^2+1)/(v_1'^2+v_2'^2)$.
	Since $u_1,\,u_1'\in\ftwom^*$, $v_1,\,v_1'\in\mu_{2^m+1}\backslash \{1\}$ and $\mu_{2^m+1}\cap\ftwom=\{1\}$, we have $u_1/u_1'=v_1'/v_1=1$, that is $u_1=u_1'$ and $v_1=v_1'$. Further, \eqref{v1v2} becomes
	\[\frac{v_1(v_2+1)}{v_1+v_2}=\frac{v_1(v_2'+1)}{v_1+v_2'},\]
	which can be simplified as
	\[v_1(v_1+1)(v_2+v_2')=0.\]
	Therefore, we can derive that $v_2=v_2'$ due to $v_1\ne 0,1$, a contradiction with $(v_1,\,v_2)\ne (v_1',\,v_2')$. Thus, $\varphi$ is an injection. In addition, $\#\mathbb{V}=2^m(2^m-1)$, which infers $\varphi$ is bijection from $\mathbb{V}$ to $\ftwon\backslash\ftwom$. Hence we can conclude that each $x\in\ftwon\backslash\ftwom$ can be uniquely written as $x=v_1(v_2+1)/(v_1+v_2)$, where $v_1,v_2\in\mu_{2^m+1}\backslash \{1\}$ and $v_1\ne v_2$. This completes the proof.
\end{proof}

\subsection{Differential-linear connectivity table}
Recently, Bar-On et al. \cite{Bar-On} presented the concept of the differential-linear connectivity table (DLCT) of $(n,m)$-functions over vector spaces. Due to the isomorphism between vector spaces and finite fields, the definition of DLCT will be converted to that of $(n,m)$-functions over finite fields.

\begin{defn}{\rm(\cite{Bar-On})}\label{def.DLCT}
Let $F$ be an $(n,m)$-function. The {\rm DLCT} of $F$ is a $2^n \times 2^m$ table whose rows correspond to input differences to $F$ and whose columns correspond to output masks of $F$. Formally, for $u\in \ftwon$ and $v\in\ftwom$, the {\rm DLCT} entry $(u,v)$ is
\[{\rm DLCT}_F(u,v)=\#\{x\in\ftwon| {\rm Tr}_1^m(v(F(x+u)+F(x)))=0\}-2^{n-1}.\]
\end{defn}
The autocorrelation of an $(n,m)$-function at point $(u,v)\in\F_{2^n}\times\F_{2^m}$ is defined as
\[{\rm AC}_F(u,v)=\sum_{x\in\F_{2^n}}(-1)^{{\rm Tr}_1^m(v(F(x+u)+F(x)))}.\]
It is known from \cite[Proposition 1]{Canteaut2} that ${\rm DLCT}_F(u,v)=\frac{1}{2} {\rm AC}_F(u,v)$, i.e.,
\begin{equation}\label{eq.DLCT}
	{\rm DLCT}_F(u,v)=\frac{1}{2}\sum_{x\in\ftwon}(-1)^{{\rm Tr}_1^m(v(F(x+u)+F(x)))}.
\end{equation}
This formula allows us to compute the {\rm DLCT} of $F$ conveniently.

It is straightforward to observe that ${\rm DLCT}_F(u,v)$ is always even, and for any $(u,v)\in\F_{2^n}\times\F_{2^m}$, $|{\rm DLCT}_F(u,v)|\leq 2^{n-1}$ with ${\rm DLCT}_F(u,v)=2^{n-1}$ when either $u=0$ or $v=0$. Therefore we only need to consider the cases where $u\in\F_{2^n}^*$ and $v\in\F_{2^m}^*$.

\begin{defn}{\rm(\cite{LLLQ})}\label{def.DLU}
	Let $F$ be an $(n,m)$-function. {The {\rm DLCT} spectrum of $F$ is defined as the multi-set
\[\Gamma_F=\{{\rm DLCT}_F(u,\,v): u\in\ftwon^*,v\in\ftwom^*\},\]
and} the differential-linear uniformity {\rm (DLU)} of $F$ is defined as
\[{\rm DLU}_{F}=\max_{u\in\ftwon^*,v\in\ftwom^*}|{\rm DLCT}_F(u,v)|.\]
\end{defn}

{The DLU and DLCT spectrum of $F$ serve as important metrics for quantifying its resistance against differential-linear cryptanalysis and related variants.} It is clear that the {\rm DLU} of any $(n,m)$-functions is upper bounded by $2^{n-1}$.
As pointed out in \cite{Canteaut2}, the lower bound on {\rm DLU} can be characterized as follows.
\begin{lem}{\rm(\cite{Canteaut2})}\label{lem.bound.DLU}
	Let $F$ be an $(n,m)$-function, where $m\geq n-1$. Then
	\[{\rm DLU}_{F}\geq\sqrt{\frac{2^{m+n+1}-2^{2n}}{4(2^m-1)}}.\]
	In special, for an $(n,n)$-function, ${\rm DLU}_{F}\geq 2^{n/2-1}+2$ if $n$ is even.
\end{lem}

Among all power functions with known DLU, the inverse function and the Kasami APN permutation are the two classes of $(n,n)$-functions with the smallest DLU for even and odd $n$, respectively, and their DLUs are given as below.

\begin{prop}{\rm(\cite{Charpin})}\label{prop-inverse}
	Let $n$ be a positive integer. The {\rm DLCT} of the inverse function $F(x)=x^{2^n-2}$ is given by
	\[{\rm DLCT}_F(u,v)=K_n(u^{-1}v)/2-1+(-1)^{{\rm Tr}_1^n(u^{-1}v)},\]
	where $u,v\in\F_{2^n}^*$. Moreover, ${\rm DLU}_F=2^{n/2}$ when $n$ is even.
\end{prop}

\begin{prop}{\rm(\cite{Gong})}\label{prop-kasami}
	Let $F(x)=x^{2^{2k}-2^k+1}$ be a power function over $\F_{2^n}$, where $n$ is odd, not divisible by $3$, and $3k\equiv \pm 1\,(\,{\rm mod} \, n)$. Then ${\rm DLU}_F=2^{(n-1)/2}$.
\end{prop}


%

\section{{\rm DLU} of some power functions}\label{cons1}
In this section, we consider the {\rm DLU} of some special power functions. For a power function $F(x)=x^d$ over $\ftwon$, where $1\le d\le 2^n-1$ is a positive integer, one can see that ${\rm DLCT}_F(u,v)={\rm DLCT}_F(1,\,u^dv)$ for all $u,\,v\in\ftwon^*$. That is to say, the {\rm DLU} of $F(x)$ is completely determined by the values of ${\rm DLCT}_F(1,\,v)$ as $v$ runs through $\ftwon^*$.

\subsection{Cubic power functions}

\begin{thm}\label{thm.cubic}
Let $F(x)=x^{2^{2k}+2^k+1}$ be a power function over $\ftwon$, where $n$ and $k$ are positive integers with $\gcd(k,\,n)=e$. Then ${\rm DLU}_F\leq 2^{(n+3e)/2-1}$ if $n$ is odd, and ${\rm DLU}_F\leq 2^{(n+4e)/2-1}$ otherwise.
\end{thm}

\begin{proof}
According to equation \eqref{eq.DLCT}, one calculates
\begin{align*}
  2{\rm DLCT}_F(1,\,v) & =\sum_{x\in\ftwon} \chi\big( {\rm Tr}_1^n\big(v\big(F(x+1)+F(x)\big)\big)\big)\\
   & =\sum_{x\in\ftwon} \chi\big( {\rm Tr}_1^n\big(v(x^{2^{2k}+2^k}+x^{2^{2k}+1}+x^{2^{k}+1}+
   x^{2^{2k}}+x^{2^{k}}+x+1)\big)\big) \\
   &=\sum_{x\in\ftwon} \chi\big( {\rm Tr}_1^n\big(vx^{2^{2k}+1}+(v^{2^{-k}}\!+v)x^{2^{k}+1}+
   (v^{2^{-2k}}\!+v^{2^{-k}}\!+v)x+v\big)\big),
\end{align*}
and then $(2{\rm DLCT}_F(1,\,v))^2$ equals
\[  \sum_{x,\,y\in\ftwon} \chi\big( {\rm Tr}_1^n\big(v(x^{2^{2k}+1}+y^{2^{2k}+1})
  +(v^{2^{-k}}\!+v)(x^{2^{k}+1}+y^{2^{k}+1})+(v^{2^{-2k}}\!+v^{2^{-k}}\!+v)(x+y)\big)\big).\]
Substituting $y$ with $x+z$ gives
\[(2{\rm DLCT}_F(1,\,v))^2=\sum_{z\in\ftwon}\chi(\phi_v(z))\sum_{x\in\ftwon}\chi(L_{v}(z)x)\]
with
\begin{equation}\label{eq.phi}
\phi_v(z)={\rm Tr}_1^n(v z^{2^{2k}+1}+ (v^{2^{-k}}+v) z^{2^k+1}+ (v^{2^{-2k}}+v^{2^{-k}}+v)z)
\end{equation}
and
\begin{equation}\label{eq.Lv}
	L_v(z)=v z^{2^{2k}}+ (v^{2^{-k}}+v) z^{2^k}+ (v^{2^{-2k}}+v^{2^{-k}}) z^{2^{-k}}+v^{2^{-2k}} z^{2^{-2k}}.
\end{equation}
Let ${\rm ker}(L_v):=\{z\in\ftwon|L_v(z)=0\}$. Then
\[(2{\rm DLCT}_F(1,\,v))^2=2^n\sum_{z\in{\rm ker}(L_v)}\chi(\phi_v(z))\leq 2^n\#{\rm ker}(L_v). \]
Next, we determine $\#{\rm ker}(L_v)$ for each $v\in\ftwon^*$. It suffices to calculate the number of solutions of the equation
\begin{align*}
  L_v(z)^{2^{2k}}&=v^{2^{2k}} z^{2^{4k}}+ (v^{2^{2k}}+v^{2^{k}}) z^{2^{3k}}+ (v^{2^{k}}+v) z^{2^{k}}+vz \\
   & =v^{2^{2k}}(z^{2^k}+z)^{2^{3k}}+v^{2^k}(z^{2^{2k}}+z)^{2^k}+v(z^{2^k}+z)=0
\end{align*}
over $\ftwon$ for any $v\in\ftwon^*$.
Set $\alpha:=z^{2^k}+z$. Then
\[L_v(z)^{2^{2k}}=v^{2^{2k}}\alpha^{2^{3k}}+v^{2^k}(\alpha^{2^{k}}+\alpha)^{2^k}+v \alpha=\beta^{2^k}+\beta=0,\]
where $\beta=v^{2^k}\alpha^{2^{2k}}+v \alpha$, which has $2^e$ solutions with respect to $\beta$ over $\ftwon$. Assume that $\beta_0$ is a solution of $\beta^{2^k}+\beta=0$. Then $v^{2^k}\alpha^{2^{2k}}+v \alpha=\beta_0$. This equation has at most $2^e$ solutions if $n$ is odd and $2^{2e}$ solutions if $n$ is even. Therefore, we can conclude that $L_v(z)=0$ has at most $2^{3e}$ (resp. $2^{4e}$) solutions in $\ftwon$ if $n$ is odd (resp. even) since $\alpha=z^{2^k}+z$ is a $2^e$-to-$1$ mapping. This implies that $\#{\rm ker}(L_v)\leq 2^{3e}$ if $n$ is odd, and $\#{\rm ker}(L_v)\leq 2^{4e}$ otherwise. Consequently, the desired result follows.
\end{proof}

For the case $e=1$, we can directly derive the following result from Theorem \ref{thm.cubic}.
\begin{cor}\label{cor1}
	Let $F(x)=x^{2^{2k}+2^k+1}$ be a power function over $\ftwon$, where $n$ and $k$ are positive integers with $\gcd(k,\,n)=1$. Then ${\rm DLU}_F\leq 2^{(n+1)/2}$ if $n$ is odd, and ${\rm DLU}_F\leq 2^{n/2+1}$ otherwise.
\end{cor}
\begin{example}\label{ex1}
	Let $k=1$ and $3\leq n\leq 18$, then $\gcd(k,n)=1$. Using Magma, it is verified that the {\rm DLU} of $F(x)=x^{7}$ over $\mathbb{F}_{2^n}$ takes exactly the values shown in Table \ref{tab-DCU-cubic}, which is consistent with Corollary \ref{cor1}.
	\begin{table}[ht]
		\caption{The {\rm DLU} of $x^{7}$ over $\F_{2^n}$ for $3\leq n\leq 18$ }\label{tab-DCU-cubic}
		\vspace{-2mm}
		\begin{center}{
				\begin{tabular}{lllllllllllllllll}
					\hline
					$n$ &  3  & 4 & 5& 6& 7& 8 & 9 & 10 & 11 & 12 & 13 & 14 &15 &16 & 17& 18\\
					${\rm DLU}_F$ & 4 & 4 & 4& 16 & 16 & 32 & 32 & 64 & 64 & 128 & 128& 256 & 256 &512&512&1024\\ \hline
			\end{tabular}}
		\end{center}
	\end{table}	
\end{example}

\begin{remark}
	Notice that in Example \ref{ex1}, when $n=4$ and $n=5$, ${\rm DLU}_F=2^{n/2}$ and ${\rm DLU}_F=2^{(n-1)/2}$ respectively. For other even and odd values of $n$, ${\rm DLU}_F=2^{n/2+1}$ and ${\rm DLU}_F=2^{(n+1)/2}$ respectively. This indicates the bound given in Corollary \ref{cor1} is tight for most cases.
\end{remark}

\begin{remark}
{For the general case $e>1$ in Theorem \ref{thm.cubic}, experiment data show that the upper bound remains tight in most cases when $e$ is small. For instance, it is achieved for $e=2$, $n=12$ with ${\rm DLU}_F=2^9$; and for $e=3$, $n=9$ with ${\rm DLU}_F=2^8$. However, for larger values of $e$, the bound is not tight in certain cases. For example, as known in \cite{Canteaut2}, ${\rm DLU}_F=2^{3k-1}$ for $n=4k$, whereas Theorem \ref{thm.cubic} only provides ${\rm DLU}_F\leq 2^{4k-1}$.
}
\end{remark}
\subsection{Dillon power functions}

\begin{thm}\label{thm.dillon}
Let $F(x)=x^{l(2^m-1)}$ be a power function over $\ftwon$, where $n=2m$ and $l$ are positive integers with $\gcd(l,2^m+1)=1$.
Then the {\rm DLCT} of $F(x)$ at point $v\in\F_{2^n}^*$ is given by
\[{\rm DLCT}_F(1,\,v)=K_m(v\bv)^2/2-2{\rm Tr}_1^n(v)K_m(v\bv)^2.\]
Moreover,
	{\begin{equation*}
		{\rm DLU}_F=\left \{\begin{array}{lll}
K_{\max}(m)^2/2, &{\rm if}\,\, \lfloor2^{m/2+1}\rfloor\equiv j\,({\rm mod}\, 4),\,\,j=0,3;\\[0.05in]
			K_{\max}(m)^2/2+2K_{\max}(m), &{\rm if} \,\, \lfloor2^{m/2+1}\rfloor\equiv j\,({\rm mod}\, 4),\,\,j=1,2,
		\end{array}\right .
	\end{equation*}
where $K_{\max}(m)$ is given by \eqref{Kmax}.}
\end{thm}

\begin{proof}
According to the equation \eqref{eq.DLCT}, one has
\[2{\rm DLCT}_F(1,\,v)  =\sum_{x\in\ftwon} \chi\big( {\rm Tr}_1^n\big(v\big((x+1)^{l(2^m-1)}+x^{l(2^m-1)}\big)\big)\big).\]
Since $F(x+1)+F(x)=1$ if $x\in\mathbb{F}_2$ and $F(x+1)+F(x)=0$ if $x\in\ftwom\backslash\mathbb{F}_2$, we have
\begin{equation}\label{eq.DLCT-dillon}
 2{\rm DLCT}_F(1,\,v)  = 2\chi({\rm Tr}_1^n(v))+2^m-2+
   \sum_{x\in\ftwon\backslash \ftwom} \chi\big( {\rm Tr}_1^n\big(v\big((x+1)^{l(2^m-1)}+x^{l(2^m-1)}\big)\big)\big).
\end{equation}
From Lemma \ref{lem.x}, let $x=v_1(v_2+1)/(v_1+v_2)$ for $x\in\ftwon\backslash\ftwom$, where $v_1\ne v_2\in\mu_{2^m+1}\backslash \{1\}$.
A calculation gives
\[x+1=\frac{v_2(v_1+1)}{v_1+v_2},\,\,x^{2^m}=\frac{v_2+1}{v_1+v_2},\,\,(x+1)^{2^m}=\frac{v_1+1}{v_1+v_2}.\]
Then \eqref{eq.DLCT-dillon} turns into
\begin{align*}
  2{\rm DLCT}_F(1,\,v) & =2\chi({\rm Tr}_1^n(v))+2^m-2+
   \sum_{v_1\ne v_2\in\mu_{2^m+1}\backslash \{1\}} \chi\big( {\rm Tr}_1^n\big(v(v_1^{-l}+v_2^{-l})\big)\big) \\
   & =2\chi({\rm Tr}_1^n(v))+2^m-2+
   \sum_{v_1\ne v_2\in\mu_{2^m+1}\backslash \{1\}} \chi\big( {\rm Tr}_1^n\big(v(v_1+v_2)\big)\big)\\
   &=2\chi({\rm Tr}_1^n(v))+2^m-2+\sum_{v_1,v_2\in\mu_{2^m+1}\backslash \{1\}} \chi\big( {\rm Tr}_1^n\big(v(v_1+v_2)\big)\big)-2^m\\
   &=2\chi({\rm Tr}_1^n(v))-2+\Big(\sum_{v_1\in\mu_{2^m+1}\backslash \{1\}} \chi\big( {\rm Tr}_1^n(vv_1)\big)\Big)^2.
\end{align*}
The second equality holds due to $\gcd(l,2^m+1)=1$.
Note that
\[\sum_{v_1\in\mu_{2^m+1}\backslash \{1\}} \chi\big( {\rm Tr}_1^n(vv_1)\big)=\sum_{v_1\in\mu_{2^m+1}} \chi\big( {\rm Tr}_1^n(vv_1)\big)-\chi({\rm Tr}_1^n(v))=1-K_m(v\bv)-\chi({\rm Tr}_1^n(v))\]
from Lemma \ref{lem.km-U}. Therefore
\[2{\rm DLCT}_F(1,\,v)=2\chi({\rm Tr}_1^n(v))-2+(1-K_m(v\overline{v})-\chi({\rm Tr}_1^n(v)))^2.\]
Then as $v$ runs through $\ftwon^*$,
\begin{equation}\label{eq.dillon.tr0}
2{\rm DLCT}_F(1,\,v)=K_m(v\bv)^2
\end{equation}
if ${\rm Tr}_1^n(v)=0$, and
\begin{equation}\label{eq.dillon.tr1}
2{\rm DLCT}_F(1,\,v)=(2-K_m(v\bv))^2-4
\end{equation}
if ${\rm Tr}_1^n(v)=1$. Therefore we conclude that
\[{\rm DLCT}_F(1,\,v)=K_m(v\bv)^2/2-2{\rm Tr}_1^n(v)K_m(v\bv).\]
Moreover, we shall analyze the value of ${\rm DLU}_F$ in two cases, based on the maximum and minimum values of $K_m(v\bv)$ as $v$ ranges over $\F_{2^n}^*$.

\textbf{Case 1}: {If $\lfloor2^{m/2+1}\rfloor\equiv j\,({\rm mod}\, 4)$ for $j=0,3$, then by Lemma \ref{lem.Kn1-2}, we have $K_{\max}(m)=4-K_{\min}(m)$. More precisely, $K_{\max}(m)=\lfloor2^{m/2+1}\rfloor$ when $j=0$, and $K_{\max}(m)=\lfloor2^{m/2+1}\rfloor+1$ when $j=3$.} In this scenario, when ${\rm Tr}_1^n(v)=0$ and $K_m(v\bv)=K_{\max}(m)$, the value of $|{\rm DLCT}_F(1,\,v)|$ reaches its maximum, which is $K_{\max}(m)^2/2$.

\textbf{Case 2}: {If $\lfloor2^{m/2+1}\rfloor\equiv j\,({\rm mod}\, 4)$ for $j=1,2$, then Lemma \ref{lem.Kn1-2} yields $K_{\max}(m)=\lfloor2^{m/2+1}\rfloor-j=-K_{\min}(m)$. This implies that when ${\rm Tr}_1^n(v)=1$ and $K_m(v\bv)=-K_{\max}(m)$, the value of $|{\rm DLCT}_F(1,\,v)|$ reaches its maximum, which is $K_{\max}(m)^2/2+2K_{\max}(m)$.}

Below, we demonstrate the existence of some $v\in\F_{2^n}^*$ that satisfies the conditions for $|{\rm DLCT}_F(1,\,v)|$ to achieve its maximum value.

{If $\lfloor2^{m/2+1}\rfloor\equiv j\,({\rm mod}\, 4)$ for $j=0,3$,} according to Lemma \ref{lem.Kn1}, there exists an element $\gamma\in\F_{2^m}^*$ such that $K_m(\gamma)=K_{\max}(m)$ as $K_m(0)=0$. Let $V:=\{v\in\F_{2^n}^*: {\rm Tr}_1^n(v)=0,\, v\bv=\gamma\}$. We proceed to show that $\#V>0$. Since $\gamma\in\F_{2^m}^*$, we assume $\gamma=\lambda^{2^m+1}$ for a fixed $\lambda\in\F_{2^n}^*$. Then $v\bv=\gamma$ implies $v=\theta \lambda$ for any $\theta\in \mu_{2^m+1}$. Hence, $V=\{\theta\in\mu_{2^m+1}: {\rm Tr}_1^n(\theta \lambda)=0\}$. A calculation yields
\begin{align*}
	\# V & =\frac{1}{2}\sum_{a\in\F_2}\sum_{\theta\in \mu_{2^m+1}}(-1)^{a {\rm Tr}_1^n(\theta \lambda)} \\
	& =\frac{1}{2}\big(2^m+1+\sum_{\theta\in \mu_{2^m+1}}(-1)^{{\rm Tr}_1^n(\theta \lambda)}\big)\\
	&=\frac{1}{2}\big(2^m+1+1-K_m(\lambda\overline{\lambda})\big).
\end{align*}
The last equality follows from Lemma \ref{lem.km-U}. Since {$K_m(\lambda\overline{\lambda})=K_m(\gamma)
=K_{\max}(m)$, and combining this with the value of $K_{\max}(m)$ given by \eqref{Kmax}}, it follows that $\# V>0$.
That is to say, there exists some $v\in\F_{2^n}^*$ such that ${\rm Tr}_1^n(v)=0$ and {$K_m(v\bv)=K_{\max}(m)$.
Thus, we conclude that ${\rm DLU}_F= K_{\max}(m)^2/2$ if $\lfloor2^{m/2+1}\rfloor\equiv j\,({\rm mod}\, 4)$ for $j=0,3$.}

{If $\lfloor2^{m/2+1}\rfloor\equiv j\,({\rm mod}\, 4)$ for $j=1,2$, according to Lemma \ref{lem.Kn1}, there exists $\gamma'\in\F_{2^m}^*$ such that $K_m(\gamma')=K_{\min}(m)=-K_{\max}(m)$. Let $V':=\{v\in\F_{2^n}^*: {\rm Tr}_1^n(v)=1,\, v\bv=\gamma'\}$. Similarly, we have $v=\theta' \lambda'$ for any $\theta'\in \mu_{2^m+1}$ and $V'=\{\theta'\in\mu_{2^m+1}: {\rm Tr}_1^n(\theta' \lambda')=1\}$, where $\gamma'=\lambda'^{2^m+1}$. Next we prove that $\#V'>0$. By calculation, we have
\begin{align*}
	\# V' & =\frac{1}{2}\sum_{a\in\F_2}\sum_{\theta'\in \mu_{2^m+1}}(-1)^{a ({\rm Tr}_1^n(\theta' \lambda')-1)} \\
	& =\frac{1}{2}\big(2^m+1-\sum_{\theta'\in \mu_{2^m+1}}(-1)^{{\rm Tr}_1^n(\theta' \lambda')}\big)\\
	&=\frac{1}{2}\big(2^m+K_m(\lambda'\overline{\lambda'})\big).
\end{align*}
Obviously $\#V'>0$ since $K_m(\lambda'\overline{\lambda'})=K_{\min}(m)=j-\lfloor2^{m/2+1}\rfloor$. In other words, there exists some $v\in\F_{2^n}^*$ such that ${\rm Tr}_1^n(v)=1$ and $K_m(v\bv)=-K_{\max}(m)$.
Consequently, we conclude that ${\rm DLU}_F=K_{\max}(m)^2/2+2K_{\max}(m)$ if $\lfloor2^{m/2+1}\rfloor\equiv j\,({\rm mod}\, 4)$ for $j=1,2$.}

This completes the proof.
\end{proof}

\begin{example}\label{ex2}
	Let $l=1$ and $n=2m$ for $2\leq m\leq 9$, then $\gcd(l,2^m+1)=1$. Using Magma, it is verified that the {\rm DLU} of $F(x)=x^{2^m-1}$ over $\mathbb{F}_{2^n}$ takes exactly the values shown in Table \ref{tab-DLCT-dillon}, which is consistent with Theorem \ref{thm.dillon}.
	\begin{table}[ht]
		\caption{The {\rm DLU} of $x^{2^{m}-1}$ over $\F_{2^{2m}}$ for $2\leq m\leq 9$ }\label{tab-DLCT-dillon}
		\vspace{-2mm}
		\begin{center}{
				\begin{tabular}{llllllllll}
					\hline
					$m$ &  2 & 3 & 4 & 5& 6& 7& 8 & 9   \\
					${\rm DLU}_F$ & 8 & 16 & 32& 72 & 128 & 240 & 512 & 1056 \\ \hline
			\end{tabular}}
		\end{center}
	\end{table}	
\end{example}

\section{{\rm DLU} of some special polynomials}\label{cons2}
In this section, we explore the {\rm DLU} of special polynomials based on certain power functions with known {\rm DLU}.
\subsection{{\rm DLU} of polynomials from quadratic functions}\label{cons2.1}
First, we investigate the DLU of a cubic function and a quadratic polynomial combined.
\begin{prop}
\label{thm.cubic+qur}
Let $G(x)=x^{2^{2k}+2^k+1}+Q(x)\in\ftwon[x]$, where $Q(x)=\sum_{0\leq i<j\leq n-1}a_{ij}x^{2^i+2^j}$ is a quadratic function and $n,\,k$ are positive integers with $\gcd(k,\,n)=e$. Then ${\rm DLU}_G\leq 2^{(n+3e)/2-1}$ if $n$ is odd and otherwise ${\rm DLU}_G\leq 2^{(n+4e)/2-1}$.
\end{prop}

\begin{proof}
According to the equation \eqref{eq.DLCT}, one has
\[2{\rm DLCT}_G(u,v) =\sum_{x\in\ftwon} \chi( {\rm Tr}_1^n(v((x+u)^{2^{2k}+2^k+1}+x^{2^{2k}+2^k+1}+Q(x+u)+Q(x))))\]
and
\begin{align*}
{\rm Tr}_1^n(v(Q(x+u)+Q(x)))&={\rm Tr}_1^n(v(\!\!\sum_{0\leq i<j\leq n-1}\!\!\!\! a_{ij}(u^{2^j}x^{2^i}+u^{2^i}x^{2^j}+u^{2^i+2^j})))\\
	& ={\rm Tr}_1^n(\!\!\sum_{0\leq i<j\leq n-1}\!\!\!\!((a_{ij}^{2^{-i}}u^{2^{j-i}}v^{2^{-i}}
+a_{ij}^{2^{-j}}u^{2^{i-j}}v^{2^{-j}})x+{a_{ij}}u^{2^i+2^j}v)).
\end{align*}
Replacing $x$ by $ux$ gives
\[2{\rm DLCT}_G(u,v) =\!\!\!\sum_{x\in\ftwon} \chi( {\rm Tr}_1^n(\theta((x+1)^{2^{2k}+2^k+1}+x^{2^{2k}+2^k+1})+L(u,v)x+C(u,v)),\]
where $\theta=u^{2^{2k}+2^k+1}v$, $L(u,v)=\sum\nolimits_{0\leq i<j\leq n-1}(a_{ij}^{2^{-i}}u^{2^{j-i}}v^{2^{-i}}+a_{ij}^{2^{-j}}u^{2^{i-j}}v^{2^{-j}})u$ and $C(u,v)=\sum_{0\leq i<j\leq n-1} {a_{ij}}u^{2^i+2^j}v$.
By a calculation, one obtains
\[2{\rm DLCT}_G(u,v)\! =\!\!\sum_{x\in\ftwon}\!\! \chi( {\rm Tr}_1^n(\theta x^{2^{2k}+1}+(\theta^{2^{-k}}\!\!+\theta)x^{2^{k}+1}+
(\theta^{2^{-2k}}\!\!+\theta^{2^{-k}}\!\!+\theta+L(u,v))x+\theta+C(u,v))),\]
and then $(2{\rm DLCT}_G(u,v))^2$ equals
\[  \sum_{x,\,y\in\ftwon}\!\!\!\! \chi\big( {\rm Tr}_1^n\big(\theta(x^{2^{2k}+1}+y^{2^{2k}+1})
+(\theta^{2^{-k}}\!\!+\theta)(x^{2^{k}+1}+y^{2^{k}+1})+(\theta^{2^{-2k}}\!\!+\theta^{2^{-k}}\!\!+\theta+L(u,v))(x+y)\big)\big).\]
Substituting $y$ with $x+z$ gives
\[(2{\rm DLCT}_G(u,v))^2=\sum_{z\in\ftwon}\chi(\phi_{\theta}(z)+{\rm Tr}_1^n(L(u,v)z))\sum_{x\in\ftwon}\chi(L_{\theta}(z)x),\]
where $\phi_{\theta}(z)$ and $L_{\theta}(z)$ are given by \eqref{eq.phi} and \eqref{eq.Lv} respectively. Then
\[(2{\rm DLCT}_G(u,v))^2=2^n\sum_{z\in{\rm ker}(L_{\theta})}\chi(\phi_{\theta}(z)+{\rm Tr}_1^n(L(u,v)z))\leq 2^n\#{\rm ker}(L_{\theta}).\]
Recall from the proof of Theorem \ref{thm.cubic} that $\#{\rm ker}(L_{\theta})\leq 2^{3e}$ if $n$ is odd and otherwise $\#{\rm ker}(L_{\theta})\leq 2^{4e}$ for ${\theta}\in \F_{2^n}^*$. Hence the desired result follows.
\end{proof}

Set $e=1$, one readily obtains the following corollary from Theorem \ref{thm.cubic+qur}.
\begin{cor}\label{cor.cubic+qur}
	Let $G(x)=x^{2^{2k}+2^k+1}+\sum_{0\leq i<j\leq n-1}a_{ij}x^{2^i+2^j}\in\ftwon[x]$, where $n$ and $k$ are positive integers with $\gcd(k,\,n)=1$. Then ${\rm DLU}_G\leq 2^{(n+1)/2}$ if $n$ is odd and otherwise ${\rm DLU}_G\leq 2^{n/2+1}$.
\end{cor}

\begin{example}\label{ex.G}
Let $k=1$ and $G(x)=x^{7}+w x^3$, where $2\leq n\leq 12$ and $w$ is a primitive element of $\F_{2^n}$. Then $\gcd(k,n)=1$. Using Magma, it is verified that the {\rm DLU} of $G(x)$ over $\mathbb{F}_{2^n}$ takes exactly the values shown in Table \ref{tab-DLCT-cubic-qur}, which is consistent with Corollary \ref{cor.cubic+qur}.	
\begin{table}[ht]
		\caption{The {\rm DLU} of $x^{7}+w x^3$ over $\F_{2^{n}}$ for $2\leq n\leq 12$ }\label{tab-DLCT-cubic-qur}
		\vspace{-2mm}
		\begin{center}{
				\begin{tabular}{lllllllllll}
					\hline
					$n$ & 3& 4 & 5 &6& 7 & 8& 9&10& 11 &12 \\
					${\rm DLU}_F$ & 4 & 8&8 & 16& 16&32 & 32 & 64 & 64 &128 \\ \hline
			\end{tabular}}
		\end{center}
\end{table}	
\end{example}


\begin{remark}
From Example \ref{ex1} and Example \ref{ex.G}, it can be observed that when $n=4$ or $n=5$, the {\rm DLU} of $x^7$ and $x^7+w x^3$ are different, which implies that they are not equivalent.
\end{remark}

{\begin{remark}
As shown in the proof of Theorem \ref{thm.cubic+qur}, including the quadratic term does not affect the upper bound on the {\rm DLU} established by the theorem. However, experimental data indicate that the quadratic term can influence the {\rm DLCT} spectrum of the function. For example, although the functions $F(x)=x^{7}$ and $G(x)=x^7+ w x^3$ over $\F_{2^8}$ have the same DLU, their DLCT spectra, as presented in Table \ref{tab-DLS}, differ.
 \begin{table}[ht]
	{\caption{The {\rm DLCT} spectrum of $F(x)=x^{7}$ and $G(x)=x^7+ w x^3$ over $\F_{2^8}$} \label{tab-DLS}
	\vspace{-2mm}
	\begin{center}{
			\begin{tabular}{llllll}
				\hline
				${\rm DLCT}_F(u,v)$ &-32 & -16 & 0 & 16 & 32  \\
				Multiplicity & 255 & 18360 & 30600 & 14790 & 1020\\ \hline
				${\rm DLCT}_G(u,v)$ &	 -32 & -16 & 0 & 16 & 32  \\
				Multiplicity &  315 & 14952 & 34536 & 14870 & 352 \\ \hline
		\end{tabular}}
	\end{center}}
\end{table}
\end{remark}}

\subsection{{\rm DLU} of polynomials from generalized cyclotomic mappings}\label{cons2.2}
Let $d,\,n$ be positive integers such that $d\mid 2^n-1$, and $\omega$ be a primitive element of $\mathbb{F}_{2^n}$. Let $C$ be the (unique) index $d$ subgroup of $\mathbb{F}_{2^n}^*$. Then the cosets of $C$ in $\mathbb{F}_{2^n}^*$ are of the form $C_i:=\omega^i C$ for $i\in\mathbb{Z}_{d}$.
It can be seen that $\mathbb{F}_{2^n}=(\bigcup_{i=0}^{d-1}C_i)\bigcup\{0\}$ and $C_i\bigcap C_j=\emptyset$ for $i\ne j$.
Let $(a_0,a_1,\cdots,a_{d-1})\in\mathbb{F}_{2^n}^{d}$ and $r_0,r_1,\cdots,r_{d-1}$ be $d$ non-negative integers.
A generalized cyclotomic mapping \cite{BW,W} of $\mathbb{F}_{2^n}$ of index $d$ is defined as follows:
\begin{equation*}
	F(x)=\left \{\begin{array}{ll}
		0, & {\rm if} \,\, x=0,\\
		a_i x^{r_i}, &{\rm if} \,\, x\in C_i,\,i\in\mathbb{Z}_{d}.
	\end{array} \right.
\end{equation*}
In special, $F(x)$ is called cyclotomic mapping if all $r_i$ are the same.
It turns out that every polynomial fixing $0$ can be represented by a cyclotomic mapping uniquely according to its index \cite{AGW2009}.
{In this subsection, we focus on the case of maximum index where each $C_i$ contains only one element.}

\begin{thm}\label{thm.cyc}
	Let $F(x)$ and $F_i(x)$ be $(n,n)$-functions, and let $\xi_i$ denote $t$ distinct elements in $\F_{2^n}$, where $i=1,\cdots,t$. Define
\begin{equation*}
	f(x)=\left \{\begin{array}{ll}
		F(x), & {\rm if} \,\, x\in\F_{2^n}\backslash N,\\
		F_i(x), &{\rm if} \,\, x=\xi_i,\,i=1,\cdots,t,
	\end{array} \right.
\end{equation*}
where $N:=\{\xi_i: i=1,\cdots,t\}$. Then ${\rm DLU}_f\leq {\rm DLU}_{F}+2t$.
\end{thm}

\begin{proof}
		According to equation \eqref{eq.DLCT}, one has
	\[2{\rm DLCT}_f(u,v) =\sum_{x\in\ftwon} \chi( {\rm Tr}_1^n(v(f(x+u)+f(x)))).\]
	Define $\Delta_f(x)=f(x+u)+f(x)$ and $N(u):=\{\xi_i,\,u+\xi_i: i=1,\cdots,t\}$. Then, distinguishing between the values of $f(x)$ for the cases $x\in N(u)$ and $x\in\F_{2^n}\backslash N(u)$, we obtain
\begin{eqnarray}
	2{\rm DLCT}_f(u,v)\!\!\!\!\!\!\!\!&&= \sum_{x\in N(u)} \chi( {\rm Tr}_1^n(v\Delta_f(x)))+ \sum_{x\in \F_{2^n}\backslash N(u)}\chi( {\rm Tr}_1^n(v\Delta_F(x))) \nonumber\\
	&&= \sum_{x\in\ftwon} \chi( {\rm Tr}_1^n(v\Delta_F(x)))+ \sum_{x\in N(u)} (\chi( {\rm Tr}_1^n(v\Delta_f(x)))-\chi( {\rm Tr}_1^n(v\Delta_F(x)))) \nonumber\\
	&&=2{\rm DLCT}_{F}(u,v)+ \sum_{x\in N(u)} (\chi( {\rm Tr}_1^n(v\Delta_f(x)))-\chi( {\rm Tr}_1^n(v\Delta_F(x)))).\label{eq-f1}
\end{eqnarray}
Note that if $\xi_i\in N(u)$, then $\xi_i+u \in N(u)$, which implies that $\# N(u)$ is even since $\xi_i+u\ne \xi_i$. Without loss of generality, we assume that $N(u)=\{\xi_{i_1},\cdots,\xi_{i_s},\xi_{i_1}+u,\cdots,\xi_{i_s}+u\}$, where $\{i_1,\cdots,i_s\}\subset\{1,\cdots,t\}$. Clearly, $\# N(u)=2s\leq 2t$.
Substituting $N(u)$ into \eqref{eq-f1}, one gets
\[2{\rm DLCT}_f(u,v)=2{\rm DLCT}_{F}(u,v)+ 2\sum_{j=1}^s (\chi( {\rm Tr}_1^n(v\Delta_f(\xi_{i_j})))-\chi( {\rm Tr}_1^n(v\Delta_F(\xi_{i_j}))))\]
due to $\Delta_f(\xi_{i_j})=\Delta_f(\xi_{i_j}+u)$ and $\Delta_F(\xi_{i_j})=\Delta_F(\xi_{i_j}+u)$. Furthermore, we obtain
\begin{equation}\label{eq-f2}{\rm DLCT}_f(u,v)={\rm DLCT}_{F}(u,v)+ \sum_{j=1}^s (\chi( {\rm Tr}_1^n(v\Delta_f(\xi_{i_j})))-\chi( {\rm Tr}_1^n(v\Delta_F(\xi_{i_j})))),
\end{equation}
which leads to
$${\rm DLU}_f\leq {\rm DLU}_F+2s\leq {\rm DLU}_F+2t.$$
This completes the proof.
\end{proof}

In the case where $N=\{\xi\}$ for a fixed element $\xi\in\F_{2^n}$, the following result can be directly obtained from Theorem \ref{thm.cyc}.
\begin{cor}\label{cor-cyc1}
Let $F(x),\,F_1(x)$ be two $(n,n)$-functions and $\xi\in\F_{2^n}$. Define
\[f(x)=\left \{\begin{array}{ll}
F(x), & {\rm if} \,\, x\in \F_{2^n}\backslash \{\xi\},\\
F_1(x), &{\rm if} \,\, x=\xi.
\end{array} \right.\]
Then ${\rm DLU}_f\leqslant{\rm DLU}_{F}+2$.
\end{cor}

By selecting $F(x)$ as the inverse function and the Kasami APN permutation, respectively, we present functions whose DLU is either identical to or closest to those of these two classes of power functions.

\begin{thm}
	Let $n=2m$ be a positive integer and $a,\,\xi\in\F_{2^n}$ satisfy $a\ne \xi^{2^n-2}$. Define
	\[f(x)=\left \{\begin{array}{ll}
		x^{2^n-2}, & {\rm if} \,\, x\in \F_{2^n}\backslash \{\xi\},\\
		a, &{\rm if} \,\, x=\xi.
	\end{array} \right.\]
The {\rm DLCT} of $f(x)$ at point $(u,v)\in\F_{2^n}^*\times \F_{2^n}^*$ is given by
\[{\rm DLCT}_f(u,v)=K_n(u^{-1}v)/2+2({\rm Tr}_1^n(v((\xi\!+\!u)^{-1}\!+\!\xi^{-1}))\!-\!{\rm Tr}_1^n(v((\xi\!+\!u)^{-1}\!+a))\!-\!{\rm Tr}_1^n(u^{-1}v)).\]
Moreover ${\rm DLU}_f\leq2^m+2$. In particular, ${\rm DLU}_f=2^m$ if $\xi= 0$.
\end{thm}

\begin{proof}
	In this case, $N(u)=\{\xi\}$ and then from \eqref{eq-f2} we can obtain
\[{\rm DLCT}_f(u,v)={\rm DLCT}_{F}(u,v)+ \chi( {\rm Tr}_1^n(v((\xi+u)^{-1}+a)))-\chi( {\rm Tr}_1^n(v((\xi+u)^{-1}+\xi^{-1}))).\]
Here $F(x)=x^{-1}$ and we define $0^{-1}=0$. From Proposition \ref{prop-inverse}, we know that
	\[{\rm DLCT}_F(u,v)=K_n(u^{-1}v)/2-1+\chi( {\rm Tr}_1^n(u^{-1}v)).\]
Using the identity $\chi({\rm Tr}_1^n(\alpha))=1-2{\rm Tr}_1^n(\alpha)$ for any $\alpha\in\F_{2^n}$, one gets
\begin{equation}\label{eq.DLCT-inv1}
	{\rm DLCT}_f(u,v)=K_n(u^{-1}v)/2+2({\rm Tr}_1^n(v((\xi\!+\!u)^{-1}\!+\!\xi^{-1}))\!-\!{\rm Tr}_1^n(v((\xi\!+\!u)^{-1}\!+a))\!-\!{\rm Tr}_1^n(u^{-1}v)).
\end{equation}
Denote $\theta_1=(\xi\!+\!u)^{-1}\!+\!\xi^{-1}$ and $\theta_2=(\xi\!+\!u)^{-1}\!+a$. Since $a\ne \xi^{-1}$, it follows that $\theta_1\ne \theta_2$. From Lemma \ref{lem.Kn1-2} and Lemma \ref{lem.Kn2}, we deduce that
the maximum and minimum values of $K_n(u^{-1}v)$ are $2^{m+1}$ and $4-2^{m+1}$, respectively, with ${\rm Tr}_1^n(u^{-1}v)=0$ and ${\rm Tr}_1^n(u^{-1}v)=1$ in each case. Observe that
$$\max_{u,v\in\F_{2^n}^*}|K_n(u^{-1}v)/2-{2 }{\rm Tr}_1^n(u^{-1}v)|=2^m$$ in either case. Combining with $|2({\rm Tr}_1^n(\theta_1 v)-{\rm Tr}_1^n(\theta_2 v))|=0$ or $2$, equation \eqref{eq.DLCT-inv1} yields
$$\max_{u,v\in\F_{2^n}^*}|{\rm DLCT}_f(u,v)|\leq 2^{m}+2.$$ Hence, we conclude that ${\rm DLU}_F\leq 2^{m}+2$.
Specially, we claim that ${\rm DLU}_F= 2^{m}$ when $\xi=0$. In the case where $\xi=0$, equation \eqref{eq.DLCT-inv1} simplifies to
\[{\rm DLCT}_f(u,v)=K_n(u^{-1}v)/2-2{\rm Tr}_1^n(u^{-1}v+av).\]
Using the value of $K_n(u^{-1}v)$, we obtain
$\max|{\rm DLCT}_f(u,v)|=2^m$ when $K_n(u^{-1}v)=2^{m+1}$ and ${\rm Tr}_1^n(u^{-1}v+av)={\rm Tr}_1^n(av)=0$ because ${\rm Tr}_1^n(u^{-1}v)=0$.
Clearly, there exist $u,v\in\F_{2^n}^*$ such that this situation occurs due to $\#\{v\in\F_{2^n}^*: {\rm Tr}_1^n(av)=0 \}>0$ and Lemma \ref{lem.Kn1}.
Thus ${\rm DLU}_F= 2^{m}$. This completes the proof.
\end{proof}

\begin{example}\label{eq.inv2}
	Let $n=8$ and $w$ be a primitive element of $\F_{2^8}$. Then the {\rm DLCT} {spectrum} of the function
	\[f(x)=\left \{\begin{array}{ll}
		x^{-1}, & {\rm if} \,\, x\in \F_{2^8}^*,\\
		w, &{\rm if} \,\, x=0
	\end{array} \right.\]
	{is given by} Table \ref{tab-DLCT-inverse1} and ${\rm DLU}_f=16$.
 \begin{table}[ht]
	\caption{The {\rm DLCT} {spectrum} of $f$ in Example \ref{eq.inv2}} \label{tab-DLCT-inverse1}
	\vspace{-2mm}
	\begin{center}{
			\begin{tabular}{llllllllll}
				\hline
				${\rm DLCT}_f(u,v)$ &-16 &  -14 & -12 & -10 & -8 & -6 & -4 & -2 & 0  \\
				Multiplicity & 1016 & 3072 & 3048 & 3328 & 5334 & 5120 & 6096  & 6144  & 4064\\ \hline
				${\rm DLCT}_f(u,v)$ &	 2 & 4 & 6 & 8 & 10 & 12 & 14 &16   \\
				Multiplicity &  4608&4572 & 4096 & 4064 & 4608 & 3556 & 1664 & 635 \\ \hline
		\end{tabular}}
	\end{center}
\end{table}	
\end{example}

\begin{remark}
	It can be verified by Magma that the {\rm DLCT} {spectrum} of the inverse function over $\F_{2^8}$ takes exactly the values as in Table \ref{tab-DLCT-inverse2}.
\begin{table}[ht]
		\caption{The {\rm DLCT} {spectrum} of $x^{-1}$ over $\F_{2^8}$ }\label{tab-DLCT-inverse2}
		\vspace{-2mm}
		\begin{center}{
				\begin{tabular}{llllllllll}
					\hline
					${\rm DLCT}_f(u,v)$ &-16 &   -12 &  -8 &  -4  & 0 & 4 &  8 &  12  &16 \\
					Multiplicity & 2040 & 6120 & 10710 & 12240 & 8160 & 9180 & 8160  & 7140  & 1275\\ \hline
			\end{tabular}}
		\end{center}
	\end{table}	
Observe that the {\rm DLCT} spectrum of $f(x)$ in Example \ref{eq.inv2} is different from that of inverse function, which infers that they are in-equivalent.
{We further compare additional cryptographic properties of the two functions in Table \ref{tab-properties}, including nonlinearity, differential uniformity (DU), and boomerang uniformity (BU). Although these indicators are identical, experimental data reveal significant differences in their Walsh spectra, differential spectra, and boomerang spectra.
}\vspace{-2mm}
	 \begin{table}[ht]
		{\caption{The cryptographic properties of $f(x)$ and $x^{-1}$ over $\F_{2^8}$ }\label{tab-properties}
		\vspace{-2mm}
		\begin{center}{
				\begin{tabular}{llllllllll}
					\hline
					Function & Nonlinearity &  DU & BU &  DLU   \\  \hline
					$f(x)$ & 112 & 4 & 6 & 16  \\ \hline
$x^{-1}$ & 112 & 4 & 6 & 16 \\ \hline
			\end{tabular}}
		\end{center}}
\vspace{-4mm}
	\end{table}
\end{remark}

From Proposition \ref{prop-kasami} and Corollary \ref{cor-cyc1}, we obtain the following result directly.
\begin{thm}\label{thm-kasami}
	Let $n$ be odd, not divisible by $3$, and $3k\equiv \pm 1\,(\,{\rm mod} \, n)$. Define
\[f(x)=\left \{\begin{array}{ll}
		x^{2^{2k}-2^k+1}, & {\rm if} \,\, x\in \F_{2^n}\backslash \{\xi\},\\
		a, &{\rm if} \,\, x=\xi,
	\end{array} \right.\]
where $a,\xi\in\F_{2^n}$ and $a\ne \xi^{2^{2k}-2^k+1}$. Then ${\rm DLU}_f\leq2^{(n-1)/2}+2$.
\end{thm}

\begin{example}\label{eq.kasami}
	Let $n=7$, $k=5$, then $n,k$ satisfy $3k\equiv 1\,(\,{\rm mod} \, n)$. Define
	\[f(x)=\left \{\begin{array}{ll}
		x^{104}, & {\rm if} \,\, x\in \F_{2^7}^*,\\
		w, &{\rm if} \,\, x=0,
	\end{array} \right.\]
	where $w$ is a primitive element of $\F_{2^7}$. Using Magma, it is verified that ${\rm DLU}_f=10$, which is consistent with Theorem \ref{thm-kasami}.
\end{example}

\begin{remark}
Based on the experimental data, we have not found parameters for which the function $f(x)$ in Theorem \ref{thm-kasami} satisfies ${\rm DLU}_f=2^{(n-1)/2}$. Whether it is possible to construct a function $f(x)$, as defined in Theorem \ref{thm.cyc}, with ${\rm DLU}_f=2^{(n-1)/2}$ by selecting appropriate branch functions deserves further investigation.
\end{remark}

{As discussed above, when $C_i=1$, the cyclotomic mappings defined in Theorem \ref{thm.cyc} can effectively construct functions with low DLU, provided that $F(x)$ is modified at only a few points. However, as the number of modified points increases, the bound on DLU may grow significantly.}
{To broaden the scope, we further explore non-trivial generalized cyclotomic mappings where $C_i$ consists of multiple elements.}
{In this regard, we present a concrete example with $d=q+1$ and $C_i=\F_q$, which leads to a function with low DLU.}

\begin{example}
{Let $w$ be a primitive element of $\F_{2^n}$, where $n=2m$ and $q=2^m$. Then the {\rm DLU} of the function
\[f(x)=\left \{\begin{array}{ll}
x^{-1}, & {\rm if} \,\, x\in \F_{2^n}\backslash\F_q,\\
w^{q+1} x^{-1}, &{\rm if} \,\, x\in\F_q
\end{array} \right.\]
takes exactly the values as in Table \ref{tab-DLCT-inverse3}.}
\vspace{-4mm}
\begin{table}[ht]
		{\caption{The {\rm DLU} of $f(x)$ over $\F_{2^{2m}}$ for $2\leq m\leq 6$ }\label{tab-DLCT-inverse3}
		\vspace{-2mm}
		\begin{center}{
				\begin{tabular}{lllllllll}
					\hline
					$m$ &  2 & 3 & 4 & 5& 6    \\
					${\rm DLU}_f$ & 8 & 12 & 28& 52 & 96 \\ \hline
			\end{tabular}}
		\end{center}}
	\end{table}	
\end{example}

It is evident that non-trivial generalized cyclotomic mappings can be employed to construct functions with low DLU. However, we have not yet discovered any functions achieving optimal or near-optimal DLU via these mappings. Identifying non-trivial generalized cyclotomic mappings that yield optimal or near-optimal DLU remains a promising direction for future research.

\section{Conclusion}\label{conc}
The concept of the DLCT was introduced by Bar-On et al. in 2019, offering a more accurate complexity analysis of the differential-linear attack. However, only a few studies have been conducted on the DLCT of cryptographic functions, especially for $(n,n)$-functions.
This paper investigated the DLCT of several infinite families of $(n,n)$-functions, including power functions and polynomials. We began by examining the DLCT of two classes of power functions: cubic functions and Dillon-type functions, and provided upper bounds for their DLU, identifying some power functions with low DLU. Next, using cubic and quadratic functions, we constructed a class of $(n,n)$-functions with low DLU.
Additionally, we leveraged the inverse function and the Kasami APN permutation, both of which have the currently known optimal DLU, to derive several classes of $(n,n)$-functions using generalized cyclotomic mappings, whose DLU is either equal to or very close to that of these functions.


\section*{Acknowledgements}
This work was supported in part by the Major Program(JD) of Hubei Province under Grant 2023BAA027, in part by the National Natural Science Foundation of China under Grant 12471492 and Grant 12271145, in part by the Innovation Group Project of the Natural Science Foundation of Hubei Province of China under Grant 2023AFA021, and in part by the Natural Sciences and Engineering Research Council of Canada under Grant RGPIN-2023-04673. The authors want to thank Jeong, Koo and Kwon for pointing out that reference \cite{Jeong}, which contains relevant results but was unintentionally omitted in the published article. As a result, we updated Table \ref{DLU-table} and references in this version.

\end{document}